%% file: data_driven_aso.tex
\documentclass[]{interact}

\usepackage{epstopdf}
\usepackage[caption=false]{subfig}

\usepackage[numbers,sort&compress]{natbib}
\bibpunct[, ]{[}{]}{,}{n}{,}{,}
\renewcommand\bibfont{\fontsize{10}{12}\selectfont}

\theoremstyle{plain}
\newtheorem{theorem}{Theorem}[section]
\newtheorem{lemma}[theorem]{Lemma}

\newtheorem{proposition}[theorem]{Proposition}

\theoremstyle{definition}

\theoremstyle{remark}

\usepackage{comment}
\usepackage{color}
\usepackage[colorlinks=true,breaklinks=true,bookmarks=true,urlcolor=blue,citecolor=blue,linkcolor=blue,bookmarksopen=false,draft=false]{hyperref}
\usepackage{algorithm}
\usepackage{algcompatible}
\usepackage{graphicx}
\usepackage{caption}
\usepackage{multirow,booktabs,dcolumn}
\usepackage[export]{adjustbox}
\usepackage[utf8]{inputenc}
\usepackage{mathrsfs}
\DeclareMathAlphabet{\mathdutchcal}{U}{dutchcal}{m}{n}
\SetMathAlphabet{\mathdutchcal}{bold}{U}{dutchcal}{b}{n}
\DeclareMathAlphabet{\mathdutchbcal}{U}{dutchcal}{b}{n}

\newcommand{\executeiffilenewer}[3]{%
	\ifnum\pdfstrcmp{\pdffilemoddate{#1}}%
	{\pdffilemoddate{#2}}>0%
	{\immediate\write18{#3}}\fi%
}

\graphicspath{{fig/}}
\begin{document}


\title{Data-driven aerodynamic shape design with distributionally robust optimization approaches}

\author{
\name{Long Chen\textsuperscript{a}\thanks{L. Chen. Email: long.chen@scicomp.uni-kl.de}, Jan Rottmayer\textsuperscript{a}\thanks{J. Rottmayer. Email: jan.rottmayer@scicomp.uni-kl.de}, Lisa Kusch\textsuperscript{a}\thanks{L. Kusch. Email: lisa.kusch@rptu.de}, Nicolas R. Gauger\textsuperscript{a}\thanks{N. R. Gauger. Email: nicolas.gauger@scicomp.uni-kl.de} and Yinyu Ye\textsuperscript{b}\thanks{Y. Ye. Email: yyye@stanford.edu}}
\affil{\textsuperscript{a}Chair for Scientific Computing, University of Kaiserslautern-Landau (RPTU), Germany;\\
\textsuperscript{b}Department of Management Science and Engineering and ICME, Stanford University, USA}
}

\maketitle

\begin{abstract}
We formulate and solve data-driven aerodynamic shape design problems with distributionally robust optimization (DRO) approaches. Building on the findings of the work \cite{gotoh2018robust}, we study the connections between a class of DRO and the Taguchi method in the context of robust design optimization. Our preliminary computational experiments on aerodynamic shape optimization in transonic turbulent flow show promising design results.
\end{abstract}

\begin{keywords}
Aerodynamic shape optimization; Distributionally robust optimization; Robust design optimization and Taguchi method; Stochastic gradient methods
\end{keywords}


\section{Introduction}

Optimization under uncertainties remains one of the major challenges for aerodynamic design optimization \cite{martins2022aerodynamic}. Uncertainty has the potential to render an optimal shape design worthless, even if obtained using sophisticated numerical approaches, as their conclusions are not realized in practice due to inevitable variations in problem data \citep{schulz2013optimal}. On the other hand, the increasing available data allows an unprecedented insight into uncertainties. For example, it is nowadays not difficult to acquire data about various key flight conditions (Mach number, altitude, temperature, etc.) \cite{brunton2021data}, but the incorporation of these data in design optimization is challenging. This paper presents a work towards data-driven design optimization using distributionally robust optimization (DRO) approaches \cite{delage2010distributionally} with applications in aerodynamic shape design.

\begin{figure}[h]
	\centering
	\begin{footnotesize}
	\executeiffilenewer{fig/multipoint.svg}{fig/multipoint.pdf}%
	{inkscape -z -C --file=fig/multipoint.svg %
		--export-pdf=fig/multipoint.pdf --export-latex}%
	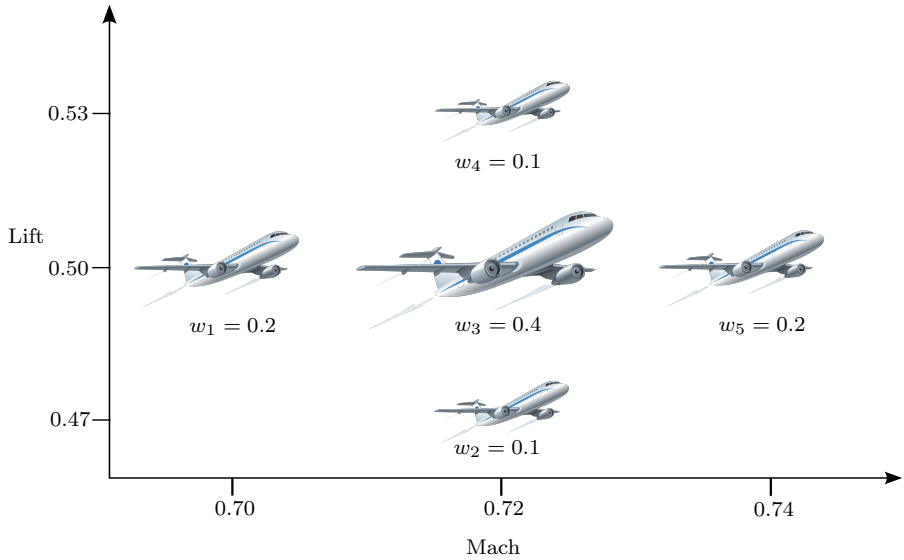%

		\caption{Multi-point aerodynamic shape optimization.}
		\label{fig:multipoint}
	\end{footnotesize}
\end{figure}

\subsection{Motivation}

To account for the variability in the flight conditions \cite{drela1998pros}, a multi-point optimization with the weighted sum formulation was introduced in \cite{reuther1997constrained}, and is considered a standard approach in academia and industry. Figure \ref{fig:multipoint} illustrates such a multi-point aerodynamic shape design optimization problem, in which the objective function is defined as a composite function through a weighted sum of five different flight conditions in Mach numbers and lift coefficients,
\begin{equation}
    f(x) = \sum_{i = 1}^{5}  \omega_i c_d \left(x; {Ma}_i, {c_l}_i \right),
\label{eq:multipoint_objective}
\end{equation}
where $c_d$ is the drag coefficient, $Ma$ is the mach number, $c_l$ is the target lift coefficient, $\omega$ is the weight factor, and $x$ is the shape design variable. The main shortcoming of \eqref{eq:multipoint_objective}
is the selection of appropriate flight points and their associated
weights \cite{nemec2004multipoint}.

\begin{figure}[h]
	\centering
	\begin{footnotesize}
	\executeiffilenewer{fig/flight_data_points_hel_muc.svg}{fig/flight_data_points_hel_muc.pdf}%
	{inkscape -z -C --file=fig/flight_data_points_hel_muc.svg %
		--export-pdf=fig/flight_data_points_hel_muc.pdf --export-latex}%
	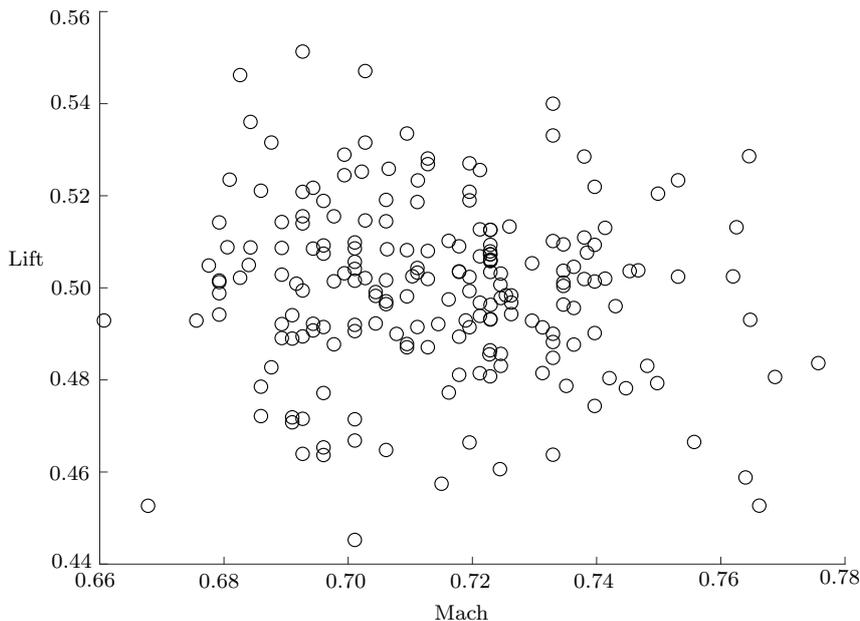%

		\caption{Flight points during cruising of an A320 from Helsinki to Munich. }
		\label{fig:data_flight_points}
	\end{footnotesize}
\end{figure}

Figure \ref{fig:data_flight_points} presents data points of a real flight of an A320 from Helsinki to Munich\footnote{$Ma$ is computed using real-world data of altitude and ground speed  (see https://www.grc.nasa.gov/www/k-12/airplane/mach.html). 
As we do not have the data for the lift force, $c_l$ is synthesized to be drawn from a normal distribution $N(0.5, 0.01)$ for demonstration purposes.}. Following the multi-point formulation, one way to incorporate these data into design optimization is to consider the following weighted sum objective function,

\begin{equation}
f(x) = \sum_{i = 1}^N \frac{1}{N} c_d(x; {Ma}_i, {c_l}_i), 
\label{eq:motivation_emo}
\end{equation}
where $N$ is the total number of data points, and $\omega_i = \frac{1}{N}, i = 1,2,..., N$. By assigning equal weight of probability to all observations, \eqref{eq:motivation_emo} becomes an empirical mean optimization problem.

If the number of data points is large, \eqref{eq:motivation_emo} would be a more realistic objective function than \eqref{eq:multipoint_objective} almost surely. On the other hand, one should realize that the empirical distribution is itself an approximation to the true distribution. For example, it is subject to shift due to measurement errors and varying conditions, such as different weather and routes. Furthermore, the design process of an aircraft (wing) precedes its mass production, and empirical data gathered (for example, from 
previous designs) cannot be the true data of the to-be-optimized design. Since an engineering solution should always be on the safe side, we want a design optimization framework that is robust to distributional changes. Such a framework is called distributionally robust optimization (DRO) \cite{delage2010distributionally}. Instead of minimizing the empirical expected value, DRO seeks an optimal design against the worst-case expected value evaluated on a set $\mathdutchcal{D}$ of candidate distributions $Q = \left[q_1, q_2,..., q_N\right]$, assuming that $\mathdutchcal{D}$ contains the true distribution $P_{true}$,
\begin{equation}
f(x) =  \max_{Q \in \mathdutchcal{D}} \sum_{i = 1}^N  q_i c_d(x;  {Ma}_i, {c_l}_i).
\label{eq:motivation_dro}
\end{equation}

This work aims to formulate, study, and solve data-driven aerodynamic shape design problems based on \eqref{eq:motivation_dro}. We summarize our main contributions as follows.

\subsection{Main contributions}

\begin{itemize}
    \item [1.] We present data-driven design optimization frameworks for aerodynamic shape design with DRO.
    \item [2.] We study the connections and differences between DRO and the classical robust design method of Taguchi.
    \item [3.] We present stochastic algorithms for solving the proposed aerodynamic shape design optimization problems that include additional inequality constraints.
    \item [4.] We report preliminary results on transonic airfoil shape design by applying our proposed optimization formulations and algorithms. 
\end{itemize}

\section{Related work}

\subsection{Aerodynamic shape optimization under uncertainties}

Probabilistic airfoil design dates back to Huyse \cite{Huyse2001}, who used a 
formulation based on minimizing the expected value of the drag coefficient. The expected drag coefficient was optimized under a uniformly distributed uncertain Mach number in a transonic Euler flow. Different studies for probabilistic airfoil design have been conducted afterwards. In \cite{Gumbert2005}, a design optimization is performed using the simplified assumption of independent
normally distributed geometrical parameters of a 3D wing. In \cite{Ong2006} geometrical and operational
uncertainties are considered separately for a 2D airfoil in Euler flow using a worst-case formulation. Lee et al. \cite{Lee2008} optimize a three-dimensional wing under uncertainties in the Mach number. The sampled mean and the variance of the drag coefficient are considered as objective functions for a multi-objective evolutionary algorithm.
In \cite{Schulz2009} different probabilistic design formulations are compared for a lift-constrained problem in a 2D Euler flow with uncertainties in the operating conditions. The Mach number and angle of attack are distributed using a truncated normal distribution. 

In the present work we consider stochastic optimization methods for aerodynamic shape optimization under uncertainties. This approach has also been pursued recently by other authors. In \cite{Geiersbach2020}, the authors present a stochastic optimization strategy for PDE-constrained optimization. In the context of aerodynamic shape optimization, stochastic gradient descent has been applied for the robust design optimization of the NACA-0012 airfoil in \cite{Jofre2022}, which considers uncertainties in operating conditions and model uncertainty. 

\subsection{Distributionally robust optimization}

Stochastic optimization is a classical framework that allows to model uncertainty within a decision-making framework. Stochastic optimization assumes that the decision maker has complete knowledge about the underlying uncertainty through a known probability distribution and minimizes a functional of the cost. Distributionally robust optimization (DRO) acknowledges that we only have partial knowledge of the statistical properties of the uncertain parameters. The core idea of DRO goes back to Scarf \cite{scarf1957min}, who proposed to define a set $\mathdutchcal{D}$ of probability distributions that is assumed to include the true distribution $P_{true}$, and reformulated the objective function with respect to the worst-case expected cost over the choice of a distribution in this set. His idea was first generalized and formally called DRO in \cite{delage2010distributionally}. This min-max DRO approach represents a systematic modeling framework that bridges the gap between data and decision-making. In the past decade, significant research has been done on DRO within the operations research and statistical learning communities (see, e.g., \cite{ben2013robust, esfahani2015data,levy2020large, wang2016likelihood}), and we refer a comprehensive review to \cite{rahimian2019distributionally}. In engineering design, DRO is a relatively new framework, with \cite{dapogny2023entropy}\cite{kapteyn2019distributionally} being the two recent literature. In this work, we formulate data-driven aerodynamic shape design problems using DRO frameworks and present practical algorithms to solve them. 




\subsection{Robust design and robust design optimization}


The Taguchi method lays the foundation of robust design and is considered the pioneering work in bringing statistics into engineering design \cite{beyer2007robust}\citep{unal1990taguchi}. A fundamental contribution of Taguchi was ``\textit{his observation that one should design a product in such a way as to make its performance insensitive to variations in variables beyond the designer's control}'' \cite{trosset1996taguchi}. Prior to the era of computational simulation and optimization, Taguchi method utilizes the design of experiment for parameter design \cite{chen1996procedure}\citep{nair1992taguchi}. With the advances in high performance computing, numerical analysis and optimization, a new field called robust design optimization (RDO) has emerged. In RDO, the mean-variance optimization (MV) is one of the most direct and fundamental mathematical optimization formulations\footnote{While mean-variance optimization is often referred to in the literature as the standard approach for robust design optimization, we would like to keep this statement open, because the influence of the Taguchi method is far-reaching \cite{tsui1992overview}. Instead, we would interpret MV as a representative mathematical optimization formulation of the Taguchi method.} following Taguchi's idea of not only optimizing the mean performance but to simultaneously reducing the variance \cite{vining1990combining}, 
\begin{equation}
\tag{MV}
\min_x  \mathbb{E}_{\xi \sim P}  \left[ f(x;\xi)\right] + \mu \mathbb{V}_{\xi \sim P} \left[ f(x; \xi) \right].
\label{eq:robust_design}
\end{equation}
In this work, building on the findings of the work \cite{gotoh2018robust} showing that MV is equivalent to $\phi$-divergence based DRO when considering a ``small amount of robustness'', we further study the connections and differences between MV and DRO in the robust design optimization context. In particular, we show that under certain conditions, the $\chi^2$-divergence based DRO can be viewed as a dual optimization formulation of the Taguchi method. 
\vskip 2mm

\section{Data-driven design optimization modeling frameworks}
\label{sec:frameworks}

Aerodynamic shape optimization is a typical PDE-constrained optimization problem, which is commonly described as
\begin{equation}
\begin{split}
&\min_{u,x} ~~~ F(x,u),  ~~~~~~~~ F: X \times U \rightarrow \mathbb{R}\\
&~\text{s.t.} ~~~~ H(x,u) = 0, ~~H: X \times U \rightarrow Z
\end{split}
\label{eq:pde_constrained_optimization}
\end{equation}
where x is the optimization variable, $u$ is the state variable, $X, U, Z$ are vector spaces, $F$ is the objective function, $H$ is the PDE constraint, and $F, H$ are sufficiently smooth functions. By implicit function theorem, if the Jacobian matrix of H about $u$ at a point $H(u_p,x_p) = 0, (u_p, x_p) \in X \times U$  is invertible, then there exists a unique continuously differentiable vector-valued function $\psi$ such that in a neighborhood of $(u_p, x_p)$ that satisfies
\begin{equation}
\begin{split}
\psi(x_p) &= u_p,\\
H(x,\psi(x)) &= 0.
\end{split}    
\end{equation}
Suppose that $\psi$ exists everywhere in the domain of the function $H$, then we can reformulate the PDE-constrained optimization \eqref{eq:pde_constrained_optimization} as an unconstrained optimization of solely the optimization variables x,
\begin{equation}
    \min_x ~~~ f(x) := F(x, \psi(x)).
\label{eq:unconstrained_pde_constrained_opt}
\end{equation}
By \eqref{eq:unconstrained_pde_constrained_opt}, in this work, we refer to a PDE-constrained optimization in the form \eqref{eq:pde_constrained_optimization} as an unconstrained optimization problem. For practical aerodynamic shape design, additional engineering constraints need to be considered. In a deterministic setting, we therefore consider a constrained optimization problem written as
\begin{equation}
\tag{$DET$}
\begin{split}
&\min_{x} ~~~ f(x) \\
&~\text{s.t.} ~~~~ g_i(x) \leq 0, ~ i = 1,2,...,m
\end{split}
\label{eq:DET}
\end{equation}
In the following, we present several data-driven optimization frameworks based on \eqref{eq:DET}.

\subsection{Empirical mean optimization for aerodynamic shape design}

We start with the formulation of empirical mean optimization (EMO) for aerodynamic shape design. Throughout this paper, we only consider the case where the data in the objective function is uncertain.
As the PDE constraint is reformulated in the objective function in \eqref{eq:unconstrained_pde_constrained_opt}, the considered model is general enough for many practical applications. 

Suppose we want to minimize the objective function $f(x; \xi)$ where $x \in \mathbb{R}^d$ is the (shape) design variable, and $\xi \in \mathbb{R}^r$ is a random variable. Assume we have observed N independent samples of $\xi_i, i = 1,2,...,N$, we want to minimize the empirical mean subject to deterministic constraints,
\begin{equation}
\tag{$EMO$}
\begin{split}
&\min_{x} ~~~ \mathbb{E}_{\xi \sim \hat{P}_N} \left[ f(x; \xi)  \right] = \sum_{i = 1}^N \frac{1}{N}  f(x; \xi_i) \\
&~\text{s.t.} ~~~~ g_i(x) \leq 0, ~ i = 1,2,...,m,
\end{split}
\label{eq:EMO}
\end{equation}
where $\hat{P}_N$ denotes the empirical distribution. By the law of large numbers, the empirical mean converges to the expected value as $N \rightarrow \infty$. Compared to the multi-point optimization \eqref{eq:multipoint_objective}, we may interpret \eqref{eq:EMO} as a multi-point optimization with a large number of points. While it is a heuristic choice of the weighing factors $\omega_i$ and data points $\xi_i$ for multi-point objective functions, \eqref{eq:EMO} minimizes the empirical risk using the known observations of $\xi$ \textemdash a fundamental principle in statistical learning. On the other hand, the statistically more meaningful optimization modeling framework of \eqref{eq:EMO} comes with the cost that the objective function is much more computationally expensive to evaluate due to the large $N$. 
To tackle the computational challenge, we rely on stochastic optimization approaches (see section \ref{sec:algorithms}), which have proven to be efficient and effective in countless machine learning applications. 

\subsection{Basics of distributionally robust optimization}

In contrast to the empirical mean optimizations, for DRO it is desired that the optimized design still has good performance under distribution shift. Specifically, DRO proposes to minimize the worst-case expected performance over a set $\mathdutchcal{D}$ of probability distributions $Q$, assuming that $\mathdutchcal{D}$ contains the true distribution $P_{true}$, 
\begin{equation}
\tag{$DRO_{con}$}
\begin{split}
&    \min_x  \Psi(x) := \max_{Q \in \mathdutchcal{D}(P, \rho)} \mathbb{E}_{\xi \sim Q} \left[ f(x; \xi)  \right], \\
& ~\textnormal{s.t.}~~  g_i(x) \leq 0, ~~ i = 1,2,...,m.
\label{eq:DRO}
\end{split}
\end{equation}
The set $\mathdutchcal{D}(P, \rho)$ is called the ambiguity set and is typically defined as
\begin{equation}
    \mathdutchcal{D}(P, \rho) := \{ Q: d(Q, P) \leq \rho \},
\end{equation}
where $d$ measures the distance between two probability distributions, $P$ is the nominal distribution, and $\rho \geq 0$ corresponds to the size (radius) of the ambiguity set.

Another well-established DRO formulation uses a soft penalty term instead of imposing a hard constrained ambiguity set, resulting in the penalized problem:
\begin{equation}
\tag{$DRO_{pen}$}
\begin{split}
& \min_x \Psi(x):= \max_Q \{\mathbb{E}_{\xi \sim Q} \left[ f(x;\xi)\right] - \delta d(Q,P)  \}, \\
& ~\textnormal{s.t.}~~  g_i(x) \leq 0, ~~ i = 1,2,...,m,
\end{split} 
\label{eq:penalized_DRO}
\end{equation}
where $\delta > 0$ is the penalty coefficient. By duality, the larger $\delta$, the less robustness is considered.

In the data-driven context, the nominal distribution is constructed from the observed data $\xi_i, i = 1,2,...,N$, i.e., we set $P = \hat{P}_N$ in \eqref{eq:DRO} and \eqref{eq:penalized_DRO}. 

\subsection{DRO frameworks based on $\phi$-divergence}
Obviously, different choice of the distance measure $d$ results in different DRO solutions. In this work, we focus on one class of such a statistical distance, the $\phi$-divergence \cite{ben2013robust}, due to its deep connections to the Taguchi method of robust design (see section \ref{sec:RD_is_DR}).

We consider $\phi$-divergence between distributions $Q$ and $P$,
\begin{equation}
    d_\phi(Q,P):= \int \phi\left( \frac{dQ}{dP}\right) dP.
\end{equation}
Suppose that the distributions $P$ and $Q$ are represented with the probability vectors, $\mathbf{p} = (p_1, ..., p_n)^T$ and $\mathbf{q} = (q_1, ..., q_n)^T$. The $\phi$-divergence of $\mathbf{q}$ from $\mathbf{p}$ is defined as
\begin{equation}
    d_\phi (\mathbf{q},\mathbf{p}) = \sum_{i = 1}^n p_i \phi(\frac{q_i}{p_i}).
    \label{eq:divergence_vec}
\end{equation}
Different choices for $\phi$ have been proposed in literature (see, e.g., \cite{pardo2018statistical}). In this work, we focus on the $\chi^2$-divergence and Kullback-Leibler (KL)/Burg entropy-based divergence and leave other useful divergences for future investigations. The function $\phi(t)$ for (modified) $\chi^2$-divergence reads
\begin{equation}
\phi_{\chi^2}(t) = \frac{1}{2}(t - 1)^2.
\end{equation}
For KL divergence, 
\begin{equation}
\phi_{kl}(t) = t \log t - t + 1.    
\end{equation}
The adjoint of KL divergence is Burg entropy-based divergence,
\begin{equation}
\phi_b(t) = - \log t + t - 1.    
\end{equation}

Figure \ref{fig:divergence_balls} illustrates different divergence balls with different radii. Note, KL divergence and Burg entropy are ``locally $\chi^2$'' \cite{polyanskiy2014lecture}, i.e., they can be well-approximated with $\chi^2$-divergence if the radius is small.

\begin{figure}[h]
\centering
\begin{footnotesize}
    \includegraphics[width=12cm]{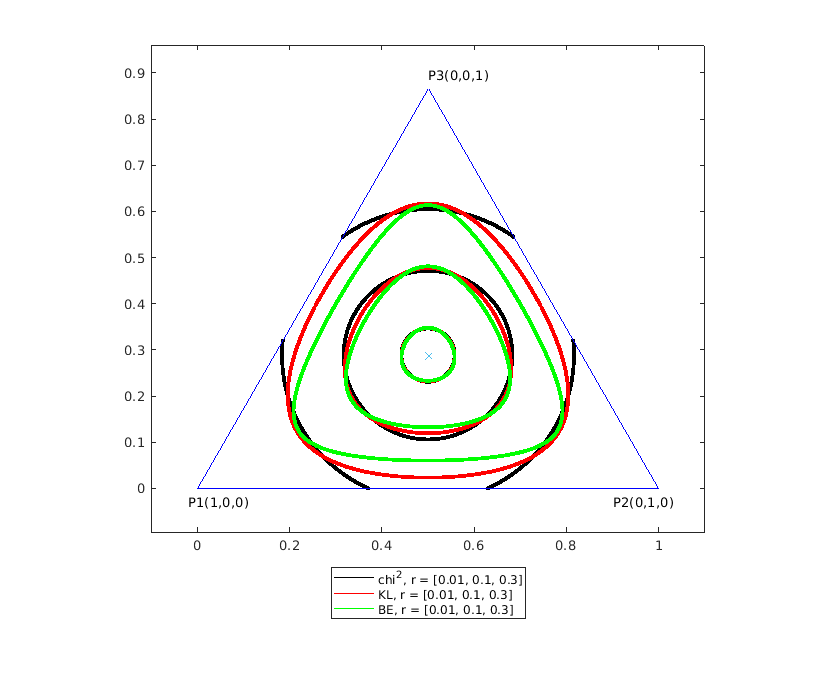}
    \caption{Ambiguity sets of different $\phi$-divergence in a barycentric coordinate system.}
    \label{fig:divergence_balls}
\end{footnotesize}
\end{figure}


\subsection{Statistically principled determination of $\mathdutchcal{D}(P, \rho)$}


DRO is a framework that naturally combines optimization with statistics. This allows for statistically principled determinations of the ambiguity set (distance measure and radius). We refer to \cite[section 6]{rahimian2019distributionally} for an overview. In the following, we give an example of the Burg entropy-based DRO. Recall, 
\begin{equation*}
\phi_b(t) = -\log t + t -1.
\end{equation*}
Consider the Burg entropy constrained ambiguity set with a uniform nominal distribution $\hat{P}_N$,
\begin{equation}
\mathdutchcal{D}(P, \rho) := \biggl\{ \mathbf{q}: \sum_{i = 1}^N \frac{1}{N} \log(\frac{1}{N q_i}) \leq \rho, \mathbf{q} \in \Delta^N \biggl\},
\label{eq:burg_entropy_set}
\end{equation}
Where $\Delta^N$ denotes the probability simplex. The inequality in the above formula is equivalent to
\begin{equation}
\sum_{i = 1}^N \log(q_i) \geq  \sum_{i=1}^N \log(\frac{1}{N}) -N \rho.
\end{equation}
\eqref{eq:burg_entropy_set} is closely related to likelihood robust optimization \cite{wang2016likelihood}, which makes assumptions that random variables taking values in set $\Xi = \{ \xi_1, \xi_2, ...,  \xi_{\tilde{n}}  \}$ and the set $\Xi$ is finite and known in advance. Assuming that N independent samples $\xi$ were observed, with $N_i$ occurrences of $\xi_i$. Then \eqref{eq:burg_entropy_set} rewrites as
\begin{equation}
\mathbb{D}(\gamma) = \biggl\{ \tilde{\mathbf{q}} = (\tilde{q}_1, \tilde{q}_2, ..., \tilde{q}_{\tilde{n}}) : \sum_{i = 1}^{\tilde{n}} N_i \log \tilde{q}_i \geq \gamma,  \tilde{\mathbf{q}} \in \Delta^{\tilde{n}} \biggr\},
\label{eq:lo_set}
\end{equation}
where $\gamma = \sum_{i = 1}^{\tilde{n}} N_i \log \frac{N_i}{N} -N \rho$, and $\mathbb{D}(\gamma)$ is called the likelihood robust distribution set. Notice that
\begin{equation}
\sum_{i = 1}^{\tilde{n}} N_i \log(\tilde{q}_i) = \log(\Pi_{i = 1}^{\tilde{n}} \tilde{q}_i^{N_i}),
\end{equation}
\eqref{eq:lo_set} implies that $\mathbb{D}(\gamma) $ contains all the distributions with support in $\Xi$ such that the observed data achieve an empirical likelihood of at least exp($\gamma$). With the set $\mathbb{D}(\gamma)$, the likelihood robust optimization writes
\begin{equation}
\min_x  \Psi(x) := \max_{\tilde{\mathbf{q}} \in \mathbb{D}(\gamma)}  \sum_{i = 1}^{\tilde{n}} \tilde{q}_i f(x; \xi_i).
\label{eq:LRO}
\end{equation} 
With Bayesian statistics and empirical likelihood theory, it was shown in \cite{wang2016likelihood} that one can approximately choose the threshold $\gamma$ as
\begin{equation}
    \gamma^* = \sum_{i = 1}^{\tilde{n}} N_i \log \frac{N_i}{N} - \frac{1}{2} \chi^2_{{\tilde{n}}-1, 1-\alpha},
\end{equation}
where $\chi^2_{{\tilde{n}}-1, 1-\alpha}$ is the $1-\alpha$ quantile of a $\chi^2$ distribution with ${\tilde{n}}-1$ degrees of freedom. 
\vskip 2mm
We conclude this section by remarking that there exist many other DRO frameworks that use different statistical distance measures for the ambiguity set, and that their applications to engineering design are worth investigating. 

\section{Connections and differences between DRO and Taguchi method}
\label{sec:RD_is_DR}

In this section, we first review recent results in the literature on the relationship between DRO and mean-variance optimization. We then show that, under certain conditions, DRO can be considered a dual optimization formulation of the Taguchi method. We also discuss the differences between the two methods when these conditions are not satisfied in order to give a complete comparison within the scope of robust design optimization.

\subsection{The equivalence of DRO and mean-variance optimization by small ambiguity set}

Let's first consider the penalized $\phi$-divergence DRO objective with nominal empirical distribution $\hat{P}_N$,
\begin{equation}
\min_x \Psi(x; \hat{P}_N):= \max_Q \{\mathbb{E}_{\xi \sim Q} \left[ f(x;\xi)\right] - \delta d_\phi(Q,\hat{P}_N)  \}.
\label{eq:DRO_objective}
\end{equation}
It was shown in \cite{gotoh2018robust} that, if $\delta$ is sufficiently large and $\phi$ satisfies the conditions,
\begin{itemize}
    \item [1.] $\phi: \mathbb{R} \rightarrow \mathbb{R} \cup \{+\infty \}$ is a closed, convex, function,
    \item [2.] $\phi(z) \geq \phi(1) = 0$ for all $z$, is twice continuously differentiable around $z = 1$ with $\phi^{\prime\prime}(1) > 0$, 
\end{itemize}
then, the penalized DRO is asymptotically equivalent to the mean-variance optimization up to $o(1/\delta)$,
\begin{equation}
\min_x \Psi(x; \hat{P}_N) = \min_x \mathbb{E}_{\hat{P}_N}  \left[ f(x;\xi)\right] + \frac{1}{2 \delta} \mathbb{V}_{\hat{P}_N} \left[ f(x; \xi) \right] + o(1/\delta).
\label{eq:equivalence_dro_mean_variance}
\end{equation}
Notice that the $\chi^2$, KL, and Burg entropy-based divergence considered in this work all satisfy these conditions. Hence, for large enough $\delta$, the penalty-based $\phi$-divergence DRO \eqref{eq:DRO_objective} well-approximates the mean-variance optimization, independent of the choice of the divergence measure.

Similar results to \eqref{eq:equivalence_dro_mean_variance} are shown for the DRO objective with hard constraint in the works in \cite{lam2016robust,namkoong2017variance}, \begin{equation}
\min_x\max_{Q \in \mathdutchcal{D}( \hat{P}_N, \rho)} \mathbb{E}_{\xi \sim Q} \left[ f(x; \xi)  \right] = \min_x \mathbb{E}_{\hat{P}_N}  \left[ f(x;\xi)\right] + \sqrt{2 \rho \mathbb{V}_{\hat{P}_N} \left[ f(x; \xi) \right]} + o(\rho).
\label{eq:dro_con_equivalence}
\end{equation}



\vskip 1mm 

With the DRO framework, we gain an important insight into the fundamental robust design framework \eqref{eq:robust_design} when $\mu$ is small: Adding variance/standard deviation to the expectation in the objective function can be seen as a robust optimization over certain distributions shifted from the nominal distribution. In the following, we provide more insight to the Taguchi method beyond the regime of ``small amount of robustness'' by studying $\chi^2$-divergence DRO and its dual formulation.



\subsection{$\chi^2$-divergence DRO: a dual view on Taguchi method}

We show non-asymptotic result on the equivalence with a particular choice of the divergence measure, the $\chi^2$-divergence. We give explicit conditions for the parameter $\delta$ and show that under these conditions, the $\chi^2$-divergence DRO can be viewed as a dual mathematical optimization formulation of the Taguchi method.

\begin{proposition}
\label{prop1}
Consider $\chi^2$-divergence distributionally robust optimization with 
\begin{equation*}
\phi_{\chi^2}(t) = \frac{1}{2}(t -1)^2.
\end{equation*}
Suppose that $\hat{P}_N$ is the empirical probability distribution on a set of independently and identically distributed data $\{ \xi_i \}_{i = 1}^N$, according to the true distribution $P_{true}$. Then, with any penalty $\delta > 0$ that satisfies
\begin{equation}
\mathbb{E}_{\hat{P}_N}\left[ f(x; \xi) \right]-f(x;\xi_k) \leq \delta, ~ \forall k,
\label{eq:delta_conditions}
\end{equation}
it holds
\begin{equation}
\max_Q \{\mathbb{E}_{Q} \left[ f(x;\xi)\right] - \delta d_{\phi_{\chi^2}}(Q,\hat{P}_N)  \} = \mathbb{E}_{\hat{P}_N}  \left[ f(x;\xi)\right] + \frac{1}{2\delta} \mathbb{V}_{\hat{P}_N} \left[ f(x; \xi) \right].
\label{eq:chi2_dro_equal_mv}
\end{equation}
\end{proposition}

\begin{proof}

~We first derive an equivalent formulation of the DRO objective with $\phi$-divergence in probability vector representation \eqref{eq:divergence_vec}. The key is to use duality theory on conjugate functions \cite{rockafellar1997convex}. The conjugate of a function $f:\mathbb{R} \rightarrow \mathbb{R} $ can be considered its dual object, and is defined as a function $f^*: \mathbb{R} \rightarrow \mathbb{R} \cup \{\infty \}$,
\begin{equation}
f^*(s) = \max_{t \geq 0} \{st - f(t) \}.
\label{eq:conjugate_f}
\end{equation}

Consider finite observations with the empirical distribution $\hat{P}_N$, the DRO objective function can be formulated with probability vectors as (we shorten the notation $q_i \geq 0, i = 1,2,...,N$ as $q\geq 0$ )
\begin{equation}
\Psi(x):= \max_{q \geq 0} \biggl\{ \sum_{i = 1}^N q_i f(x;\xi_i) - \delta \sum_{i = 1}^N p_i \phi(\frac{q_i}{p_i})  \biggr\},
\label{eq:dro_obj_1}
\end{equation}
with the constraint
\begin{equation}
h(q) = \sum_{i = 1}^N q_i - 1 = 0.
\label{eq:dro_obj_2}
\end{equation}
Introducing Lagrange Multiplier $\lambda$, we can rewrite the constrained maximization problem \eqref{eq:dro_obj_1}-\eqref{eq:dro_obj_2} as
\begin{equation}
\begin{split}
\Psi(x) &= \min_{\lambda \in \mathbb{R}} \max_{q \geq 0} \biggl\{ \sum_{i = 1}^N q_i f(x;\xi_i) - \delta \sum_{i = 1}^N p_i \phi(\frac{q_i}{p_i})  \biggr\} - \lambda \left( \sum_{i = 1}^N q_i - 1 \right) \\ 
&= \min_{\lambda \in \mathbb{R}} \max_{q \geq 0} \sum_{i = 1}^N \left[ q_i f(x;\xi_i)  - \delta p_i \phi(\frac{q_i}{p_i}) -\lambda q_i  \right] + \lambda.
\end{split}
\end{equation}
Since $\hat{P}_N$ is the empirical distribution, we have $p_i = \frac{1}{N}, i = 1,...,N$. By the definition of conjugate function \eqref{eq:conjugate_f}, it then follows,
\begin{equation}
\begin{split}
\Psi(x)  & = \min_{\lambda \in \mathbb{R}} \max_{q \geq 0} \sum_{i = 1}^N \left[ q_i f(x;\xi_i)  - \delta \frac{1}{N} \phi(N q_i)  -\lambda q_i  \right] + \lambda  \\
& = \min_{\lambda \in \mathbb{R}} \max_{q \geq 0} \delta \frac{1}{N}\sum_{i = 1}^N \left[ \left( \frac{f(x; \xi_i) - \lambda}{\delta} \right) N q_i - \phi(N q_i)  \right] + \lambda \\
& = \min_{\lambda \in \mathbb{R}} \frac{1}{N}\sum_{i = 1}^N \left[\delta  \phi^*  \left( \frac{f(x; \xi_i) - \lambda}{\delta} \right) + \lambda \right].
\end{split}
\end{equation}
Thus, the equivalent dual DRO objective with penalty formulation writes
\begin{equation}
    \Psi(x) = \min_{\lambda \in \mathbb{R}} \mathbb{E}_{\hat{P}_N} \left[\delta  \phi^*  \left( \frac{f(x; \xi) - \lambda}{\delta} \right) + \lambda \right].
\label{eq:equivalent_DRO_dual}    
\end{equation}
Recall the (modified) $\chi^2$-divergence, 
\begin{equation*}
\phi_{\chi^2}(t) = \frac{1}{2}(t - 1)^2,
\end{equation*}
whose conjugate function writes
\begin{equation}
\phi_{\chi^2}^*(s) =     
\left\{
\begin{split}
& -\frac{1}{2}    ~~~~~~~~~      s  < -1,   \\
& \frac{1}{2}s^2 + s, ~~~~ s \geq -1.\\
\end{split}
\right.
\label{eq:chi^2_div_conjugate}
\end{equation}
Having 
$s_k = \frac{f(x;\xi_k) - \lambda}{\delta}, $ and suppose that $s_k \geq -1$ holds for every $k = 1, 2,..., N$,
then
\begin{equation}
\begin{split}
    \Psi(x) &=  \min_{\lambda \in \mathbb{R}} \mathbb{E}_{\hat{P}_N} \delta \left[ \frac{1}{2} \left(\frac{f(x;\xi) - \lambda}{\delta} \right)^2 + \left(\frac{f(x;\xi) - \lambda}{\delta} \right) \right]  + \lambda \\
    & = \min_{\lambda \in \mathbb{R}} \mathbb{E}_{\hat{P}_N} \left[  \frac{1}{2\delta} \left( f(x;\xi) - \lambda \right)^2 + f(x;\xi) \right].
\end{split}    
\label{eq:dro_dual_chi2}
\end{equation}
The minimizer $\lambda^\star$ for \eqref{eq:dro_dual_chi2} is the arithmetic mean of $f(x; \xi_i), i = 1,..., N $, i.e., $\lambda^\star = \mathbb{E}_{\hat{P}_N } \left[ f(x; \xi) \right]$. Thus, 
\begin{equation}
\Psi(x) = \frac{1}{2\delta}\mathbb{E}_{\hat{P}_N } \left[ (f(x;\xi) - \mathbb{E}_{\hat{P}_N } \left[ f(x; \xi) \right] )^2 \right] + \mathbb{E}_{\hat{P}_N } \left[ f(x; \xi)\right],
\label{eq:dro_equal_mv}
\end{equation}
which is the mean-variance optimization \eqref{eq:robust_design} objective function. Hence our conclusion: $\chi^2$-divergence based DRO is equivalent to the mean-variance optimization when $\mathbb{E}_{\hat{P}_N}\left[ f(x; \xi) \right]-f(x;\xi_k) \leq \delta, k = 1,2,...,n $.  
\end{proof}

A similar result can be shown for DRO with hard constraint formulation using the same procedure as in the proof of Proposition \ref{prop1}.

\begin{lemma}
\label{lemma1}
Consider $\chi^2$-divergence distributionally robust optimization with 
\begin{equation*}
\phi_{\chi^2}(t) = \frac{1}{2}(t -1)^2.
\end{equation*}
Suppose that $\hat{P}_N$ is the empirical probability distribution on a set of independently and identically distributed data $\{ \xi_i \}_{i = 1}^N$, according to the true distribution $P_{true}$. Then, with any parameter $\rho > 0$ that satisfies
\begin{equation}
\left( \mathbb{E}_{\hat{P}_N}\left[ f(x; \xi) \right] - f(x;\xi_k) \right) \sqrt{2\rho} \leq \sqrt{\mathbb{V}_{\hat{P}_N}\left[ f(x; \xi) \right]}, ~~ \forall k, 
\label{eq:rho_conditions}
\end{equation}
it holds
\begin{equation}
\max_{Q \in \mathdutchcal{D}_{\chi^2}( \hat{P}_N ,\rho)} \mathbb{E}_{\xi \sim Q} \left[ f(x; \xi)  \right] = \mathbb{E}_{\hat{P}_N}  \left[ f(x;\xi)\right] + \sqrt{2 \rho \mathbb{V}_{\hat{P}_N} \left[ f(x; \xi) \right]}.
\label{eq:chi2_dro_equal_std}
\end{equation}
\end{lemma}

\begin{proof}
Introducing Lagrange Multiplier $\lambda, \nu$, we can rewrite the hard constraint DRO objective as
\begin{equation}
\begin{split}
\Psi(x) &= \min_{\lambda \in \mathbb{R}, \nu \geq 0} \max_{q \geq 0}  \sum_{i = 1}^N q_i f(x;\xi_i)  - \lambda \left( \sum_{i = 1}^N q_i - 1 \right) - \nu \left(\sum_{i = 1}^N \frac{1}{N} \phi(N q_i) - \rho \right)\\ 
&= \min_{\lambda \in \mathbb{R}, \nu \geq 0} \max_{q \geq 0} \frac{\nu}{N} \sum_{i = 1}^N \left[ \left( \frac{ f(x;\xi_i)  -\lambda}{\nu} \right) N q_i  - \phi(N q_i) \right] + \lambda + \nu \rho \\
& = \min_{\lambda \in \mathbb{R}, \nu \geq 0} \frac{\nu}{N}  \sum_{i = 1}^N \phi^* \left( \frac{ f(x;\xi_i)  -\lambda}{\nu} \right) + \lambda + \nu \rho. 
\end{split}
\end{equation}
Hence,
\begin{equation}
\Psi(x) = \min_{\lambda \in \mathbb{R}, \nu \geq 0} \mathbb{E}_{\hat{P}_N} \left[\nu  \phi^*  \left( \frac{f(x; \xi) - \lambda}{\nu} \right) + \lambda + \nu \rho \right].
\end{equation}
Suppose that $ \left( \frac{f(x; \xi_k) - \lambda}{\nu} \right) \geq -1 $ holds for every $k = 1, 2,..., N$, then with $\chi^2$-divergence we have
\begin{equation}
\Psi(x)  = \min_{\lambda \in \mathbb{R}, \nu \geq 0} \mathbb{E}_{\hat{P}_N} \left[ \frac{1}{2\nu}(f(x;\xi) - \lambda)^2 + f(x; \xi) + \nu \rho \right].
\label{eq:dual_DRO_con}
\end{equation}
The right hand side of \eqref{eq:dual_DRO_con} attains the minimum at
\begin{equation}
\begin{split}
\lambda^\star &= \mathbb{E}_{\hat{P}_N } \left[ f(x; \xi) \right], \\
\nu^\star & = \sqrt{\frac{\mathbb{V}_{\hat{P}_N}\left[ f(x; \xi) \right]}{2\rho}}.
\end{split}    
\end{equation}
Finally,
\begin{equation}
\Psi(x)  =  \mathbb{E}_{\hat{P}_N} \left[f(x; \xi) \right] + \sqrt{2\rho \mathbb{V}_{\hat{P}_N}\left[f(x; \xi) \right] }.
\end{equation}
\end{proof}

Comparing \eqref{eq:dro_obj_1} -\eqref{eq:dro_obj_2} with \eqref{eq:equivalent_DRO_dual}, the former maximization problem over distributions in the DRO objective is the dual of the latter minimization problem over the Lagrange multiplier $\lambda$. When using the $\chi^2$-divergence, the analytical solution to \eqref{eq:equivalent_DRO_dual} is the mean-variance objective function. Therefore, under conditions \eqref{eq:delta_conditions}, the $\chi^2$-divergence DRO objective is the dual of the mean-variance objective. The dual view on the Taguchi method with DRO allows new possibilities in the modeling of robust engineering design problems:
\begin{itemize}
    \item [1.] If $f(x; \xi)$ is convex in $x$, then the DRO objectives are convex surrogates for the variance/standard deviation regularized objectives, which are generally non-convex even if $f$ is convex in $x$. Therefore, if the underlying design optimization problem is convex (or can be equivalently formulated as a convex problem), an equivalent DRO formulation is easier to solve (see \cite{namkoong2017variance}).
    \item [2.] DRO allows statistically principled and data-driven analysis and determination of the robust parameter $\delta$ and $\rho$, which are good suggestions for the factors in mean-variance/standard deviation formulations. Thus, it is an immediate contribution to existing robust design frameworks and algorithms.
    \item [3.] For sufficiently large robustness (i.e., $\delta$ small enough or $\rho$ large enough), different $\phi$-divergence results in substantially different ambiguity sets (see figure \ref{fig:divergence_balls}). In such situations, different $\phi$-divergence DROs will generally give different design solutions. Intuitively, one may think that different $\phi$-divergences lead to ``mean''-``variance'' (-``standard deviation'') optimizations, in which the statistical moments are differently measured. 
\end{itemize}

\subsection{Differences between DRO and mean-variance optimization}

Proposition \ref{prop1} states that if $\delta$ satisfies \eqref{eq:delta_conditions}, $\chi^2$-divergence DRO is the dual formulation of the mean-variance optimization. In the following, we study the case when \eqref{eq:delta_conditions} cannot be satisfied and elaborate on the differences between both optimization frameworks. 

Suppose that $\delta$ is small enough such that $s_k = \frac{f(x;\xi_k) - \lambda^\star}{\delta} < -1 $ for some sample $\xi_k$, where $\lambda^\star$ is the minimizer of \eqref{eq:equivalent_DRO_dual}. Then, $\phi_{\chi^2}^*(s_k) = -\frac{1}{2}$ by the conjugate function \eqref{eq:chi^2_div_conjugate}. As a result, the equivalence of DRO and MV shown in \eqref{eq:chi2_dro_equal_mv} no longer holds. When computing the gradient of the DRO objective, the following differences appear when compared to the equivalent MV objective \eqref{eq:dro_equal_mv}:

\begin{itemize}
    \item [1.] the minimizer $\lambda^\star \neq \mathbb{E}_{\hat{P}_n } \left[ f(x; \xi) \right]$ in general;
    \item [2.] the contribution of the gradient $\nabla f(x; \xi_k)$ vanishes in $\nabla \Psi(x)$.
\end{itemize}

To gain some intuition of these differences, we consider an illustrative example with the objective function (\cite[equation (28)]{beyer2007robust})
\begin{equation}
\ell(x; \xi) = \xi - (\xi - 1)x^2,
\label{eq:example_ell}
\end{equation}
where $0 \leq x \leq 1$, and $\xi$ has the support $\Xi = \{ -0.8, -0.4, 0.0, 0.4, 0.8 \}$. Furthermore, we assume $\xi$ has a nominal empirical distribution $\hat{P} = \left[ \frac{1}{5}, \frac{1}{5}, \frac{1}{5}, \frac{1}{5}, \frac{1}{5} \right]$. From figure \ref{fig:dro_ro_mv_comparisons}, it is clear that,
\begin{equation}
\begin{split}
\mathop{\text{argmin}}_{ 0 \leq x \leq  1} ~ \mathbb{E}_{\hat{P}} \left[ \ell(x; \xi) \right] = 0, \\
\mathop{\text{argmin}}_{ 0 \leq x \leq  1} ~ \mathbb{V}_{\hat{P}} \left[ \ell(x; \xi) \right] = 1. \\
\end{split}
\end{equation}
Consider the mean-variance optimization for $\ell(x; \xi)$ with $\mu \geq 0$,
\begin{equation}
\min_{ 0 \leq x \leq  1} \mathbb{E}_{\hat{P}}  \left[ \ell(x;\xi)\right] + \mu \mathbb{V}_{\hat{P}} \left[ \ell(x; \xi) \right].
\label{eq:mv_example}
\end{equation}
The solution of the above bi-objective optimization lies in the Pareto frontier $ \{x \in \mathbb{R}: 0 \leq  x \leq 1 \}$, depending on the choice of $\mu$. 

\begin{figure}[h]
\centering
\begin{footnotesize}
    \includegraphics[width=9cm]{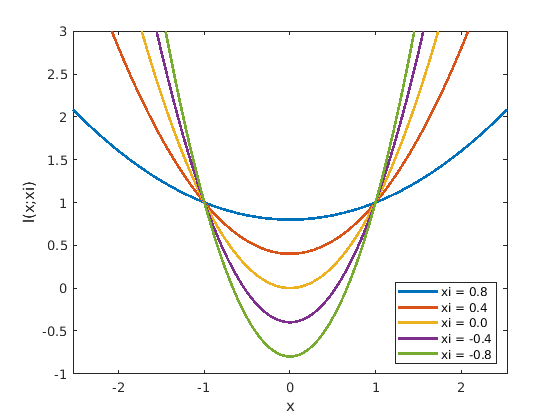}
    \caption{Function plots for \eqref{eq:example_ell} with different $\xi$.}
    \label{fig:dro_ro_mv_comparisons}
\end{footnotesize}
\end{figure}

Let's now consider the $\phi$-divergence DRO formulation for \eqref{eq:example_ell},
\begin{equation}
\min_{ 0 \leq x \leq  1} \Psi(x) = \max_{q \geq 0} \biggl\{ \sum_{i = 1}^5 q_i \ell(x;\xi_i) - \delta \sum_{i = 1}^5 p_i \phi(\frac{q_i}{p_i})  \biggr\}.
\label{eq:dro_obj_example}
\end{equation}
Let $q^\star(x) = \left[q_1^\star(x), q_2^\star(x), ...,  q_5^\star(x) \right]$ be the maximizer of the inner problem at $x$, we can write the gradient of the DRO objective \eqref{eq:dro_obj_example} as
\begin{equation}
\nabla \Psi(x) = \sum_{i = 1}^5 q_i^\star(x) \nabla \ell(x;\xi_i),
\end{equation}
which is a convex combination of all sample gradients. By a gradient descent procedure, we can see that
\begin{equation}
\mathop{\text{argmin}}_{ 0 \leq x \leq  1} ~ \Psi(x) = 0, 
\label{eq:example_dro_res}
\end{equation}
and the result hold for any $\delta \geq 0$. The result \eqref{eq:example_dro_res} makes sense, because the solution of DRO is upper bounded by the solution of the related robust optimization,
\begin{equation}
\min_{ 0 \leq x \leq  1} \max_{\xi \in \Xi} ~~\xi - (\xi - 1)x^2,
\end{equation}
which seeks a robust solution over the supports instead of the distributions. Obviously, $\xi = 0.8$ maximizes the inner optimization problem (notice that $0 \leq x \leq 1$), and the minimizer for $\ell(x; \xi = 0.8)$ is $x^\star = 0$. 

\vskip 1mm

To summarize, while mean-variance optimization is a bi-objective optimization, the min-max formulation of DRO results in solutions that are upper bounded by the related robust optimization. As a consequence, DRO formulation provides a degree of certainty to avoid generating overly conservative designs. On the other hand, depending on applications, the whole mean-variance trade-off Pareto front is worthwhile investigating.

\section{Algorithms}
\label{sec:algorithms}

The major challenges in applying the present data-driven optimization modeling frameworks to aerodynamic shape design include 1) the linear scaling of the number of function/gradient evaluations with the number of data points, and 2) the high computational cost in solving even a single-point CFD computation. If there is a surrogate model at hand that provides cheap (and accurate enough) evaluations of $f(x; \xi)$ for any feasible $x$ and $\xi$, then one can resort to the empirical mean formulation \eqref{eq:EMO} or DRO formulations \eqref{eq:DRO} and \eqref{eq:penalized_DRO} and solve these data-driven optimization problems with well-established deterministic optimization solvers. However, training such a surrogate for practical aerodynamic shape design problems is yet an ongoing endeavor, and it is especially demanding if the number of shape variable $x$ and the dimension of the uncertain parameter $\xi$ is large. In this work, instead, we solve the present optimization formulations using stochastic gradient approaches and make use of efficient gradient computations via the adjoint method and algorithmic differentiation. 

\subsection{Stochastic gradient methods for EMO and DRO}

Perhaps the most well-known method to solve the unconstrained empirical mean optimization problem is the stochastic gradient descent (SGD). Stochastic gradient methods solve EMOs with a number of gradient computations independent of the total number of data points $\xi$ and the dimension of the design variable $x$. If we only solve unconstrained EMOs, we can indeed rely on the rich optimization algorithms for stochastic optimization with SGD and its advanced variants.

Solving DRO with stochastic gradient methods is more challenging because the inner maximization over candidate distributions makes cheap sampling-based gradient estimates biased (i.e., it exists approximation error). The work \cite{levy2020large} shows that the bias upper bound is $O(\sqrt{1/n})$ and $O(1/n)$ for $\chi^2$ constraint and $\chi^2$ penalty formulation, respectively, where $n$ is the batch size. Despite the somewhat pessimistic error bound, the same work reported promising computational experiments that the error floor becomes negligible even for batch sizes as small as 10. The work \cite{ghosh2018efficient} considers the same unconstrained DRO setting and suggests dynamically increasing the batch sizes to balance the accuracy and computational efficiency. For practical aerodynamic shape design, a large batch size is often prohibitively expensive. Under a constrained computational budget, we choose a fixed batch size $n$ for our algorithms. 

\subsection{Dealing with additional inequality constraints}

One challenge left is to deal with the additional inequality constraints in \eqref{eq:EMO}, \eqref{eq:DRO} and \eqref{eq:penalized_DRO}. Recently, some of the authors have proposed a gradient method for inequality constrained optimization (GDAM) \cite{chen2023gradient}. In this work, we extend it for the setting of stochastic optimization. 

\subsubsection{GDAM in the deterministic setting}
\label{sec:GDAM}

The method can be motivated with the interior-point method for constrained optimization \cite{nocedal1999numerical}. Consider the deterministic constrained optimization problem
\eqref{eq:DET}, the logarithmic barrier function writes
\begin{equation}
\Phi(x)= - \sum_{i=1}^m \log (-g_i(x)).
\end{equation}
The classical barrier interior-point method solves a sequence of $\eta$-subproblems,
\begin{equation}
\min_x ~~  f(x) + \eta \Phi(x),
\label{eq:logarithmic_barrier}
\end{equation}
where $\eta > $ is the barrier parameter. As $\eta \rightarrow 0$, the solution of \eqref{eq:logarithmic_barrier} converges to the original problem \eqref{eq:DET}. Following the double-loop structure of second-order IPM, one can design a gradient method: 1) iteratively decrease the barrier parameter $\eta$ in the outer loop, and 2) solve the subproblem \eqref{eq:logarithmic_barrier} with the negative gradient direction
\begin{equation}
    d_{\Phi,\eta} = - \nabla f(x) - \eta \nabla \Phi(x).
\label{eq:gd_barrier}    
\end{equation}
In contrast to the double-loop method, GDAM suggests a descent direction of the objective $f(x)$ as
\begin{equation}
\tag{GDAM}
\mathbf{s}_\zeta (x) = -\frac{\nabla f(x)}{\| \nabla f(x) \|}
-\zeta\frac{\nabla \Phi(x)}{\| \nabla \Phi(x) \|}, ~~ \zeta \in [0,1),
\label{eq:GDAM}
\end{equation}
where $\| \cdot \|$ denotes L2-norm. Scaling GDAM's search direction $s_\zeta$ with $\| \nabla f(x_k) \|$, we obtain
\begin{equation}
\|\nabla f(x) \| s_\zeta = -\nabla f(x) - \zeta \frac{\| \nabla f(x) \|}{\| \nabla \Phi(x) \|} \nabla \Phi(x).
\label{eq:gdam_barrier_connection}
\end{equation}
Comparing \eqref{eq:gdam_barrier_connection} with \eqref{eq:gd_barrier}, it is then clear that the (scaled) GDAM suggests a \textit{dynamic} computation of the barrier parameter at \textit{each step} $k$,
\begin{equation}
 \eta(x_k) =\zeta \frac{\| \nabla f(x_k) \|}{\| \nabla \Phi(x_k) \| }.
 \label{eq:barrier_zeta}
\end{equation}
It was shown in \cite{chen2023gradient} that GDAM is a computationally practical method for shape optimizations, which motivated us to its extension for solving the present data-driven shape design problems. 

\subsubsection{GDAM in the stochastic setting}
\label{sec:GDAMstoch}

For stochastic optimizations, solving $\eta$-barrier subproblems with stochastic gradient methods (i.e., with a stochastic gradient estimation for \eqref{eq:gd_barrier}) has been used in machine learning applications when additional constraints are considered, see, e.g., \cite{kervadec2022constrained}\cite{liu2020ipo}. The duality gap associated with primal $x^\star$ of the $\eta$-barrier problem and implicit dual feasible $\lambda^\star$ of the original inequality constrained problem \eqref{eq:DET} is upper bounded by $m \eta$. However, a small $\eta$ leads to high gradient oscillations near the boundary of the feasible set (in both the deterministic and stochastic settings).

In this work, we extend GDAM to the stochastic setting. More precisely, the gradient evaluation of the objective function is considered to be noisy. In the deterministic setting, GDAM results in optimization trajectories that travel alongside the central path within a $\Theta_\zeta$ neighborhood defined by the hyperparameter $\zeta$ (a larger $\zeta$ indicates a smaller cone neighborhood, see \cite[Lemma 5.6]{chen2023gradient}). In the stochastic setting, however, the central path is no longer unique and varies with the gradient noise. To ensure progress in reducing the objective function value, we relax the parameter $\zeta = \tau \zeta, 0 < \tau < 1$, when the optimization progress stalls. Intuitively, we allow a larger central path neighborhood to bound a number of central paths resulting from the imprecise gradient information. We show a pseudocode of stochastic GDAM in Algorithm \ref{alg:stogdam}, which uses a fix-length step size rule, i.e., $\| x_{k+1} - x_k \| = \beta.$ Furthermore, the algorithm can be equipped with advanced stochastic optimization methods, such as ADAM \cite{kingma2014adam}, to accelerate the optimization.

For aerodynamic shape design problems, our observation is that a design should approach the boundary of the feasible set close enough so as to achieve good performance in the objective function. In the next section, we report computational experiments and show that stochastic GDAM achieves fairly good design solutions at the constraint boundaries. 


\begin{algorithm}
\caption{Stochastic GDAM}\label{alg:stogdam}
\begin{algorithmic}[1]
\State \textbf{Initialization:} $x_0$, $\zeta$, $\tau$, $\beta$, $n$, $k = 0$,
\WHILE{Stopping criteria is not met}
\State Sample $\xi_1, ..., \xi_n$ and compute $\nabla f(x_k; \xi_1), ..., \nabla f(x_k; \xi_n)$.
\State Compute $g_i(x_k), i = 1,...,m$.
\State Compute $\Tilde{s}_\zeta(x_k) = \frac{1}{n}\sum_{i=1}^n \mathbf{s}_\zeta(x_k; \xi_i)$.
\State \textbf{Update:} 
$\alpha_k \leftarrow \frac{\beta}{\| \Tilde{s}_\zeta (x_k) \|}$,
$x_{k+1}  \leftarrow x_k + \alpha_k  \Tilde{s}_\zeta (x_k),$ 
$k \leftarrow k+1$
\State \textbf{Update:} $\zeta \leftarrow \tau \zeta$, if $\mathbb{E}_{\hat{P}_n} \left[ f(x_{k+1};\xi) \right] > \mathbb{E}_{\hat{P}_n} \left[ f(x_{k};\xi) \right] $
\State \textbf{(Line search to update $\beta$)}
\ENDWHILE
\end{algorithmic}
\end{algorithm}

\section{Preliminary computational examples}

We apply the above presented concepts frameworks and algorithms to aerodynamic shape design optimization of the RAE2822 airfoil in transonic turbulent flow 
with an uncertain free-stream Mach number that assumes to have a ``true'' distribution as a normal distribution $N(\mu, \sigma)$, where $\mu = 0.729$ and $\sigma = 0.01$. We draw samples from this distribution to construct an empirical distributions $\hat{P}_N$ to facilitate the experiment. In the following, we consider the deterministic aerodynamic shape optimization problem, a probabilistic approach based on the empirical mean and a distributionally robust design approach. 

\subsection{Deterministic aerodynamic shape optimization of RAE2822 in transonic turbulent flow}

First, we consider the deterministic shape optimization problem based on the mean Mach number, i.e., $\bar{\xi}=0.729$ (which we refer to as the nominal Mach number in the following). Furthermore, the lift coefficient is fixed to a value of $0.724$. This is directly done in the CFD solver by adjusting the angle of attack in an iterative fashion. The optimization goal is to minimize the drag coefficient $c_d$ while constraining the the pitching moment $c_m$ and the maximum thickness $t$ of the airfoil. The corresponding optimization problem reads
\begin{equation}
    \begin{split}
        &\min_{x} ~~c_d(x)  \\
        &~~\text{s.t.} ~~~ c_m(x) \leq 0.093, t(x) \geq 0.12.
    \end{split}
    \label{eq:detAero}
\end{equation}

The underlying flow equations are modeled using the Reynolds-averaged Navier-Stokes equations with a Spalart-Allmaras turbulence model and are solved with the help of the open-source software SU2 \cite{Economon2016}. The solution method is a finite volume method based on a vertex-based median-dual grid using a Jameson-Schmidt-Turkel scheme for the spatial discretization. The design variables $x$ are composed of $38$ Hicks-Henne parameters describing the shape. More specifically, SU2 enables the provision of sensitivites of the objective function with respect to the design variables with the help of the discrete adjoint method based on a fixed-point formulation \cite{Sagebaum2019}. The software suite also directly provides a Python framework for design optimization that also enables the integration of optimizers like GDAM. 

In a first experiment, the integrated Python optimizer SLSQP \cite{Kraft1988} and GDAM are both applied to solve problem \eqref{eq:detAero}. We compare the results of GDAM to the results of SLSQP since SLSQP is commonly applied in the context of constrained aerodynamic shape optimization problems based on SU2. 

GDAM converges to a design depicted in Figure \ref{fig:det_design} with a reduction of the drag coefficient of the initial design $c_d = 0.013462$ to a drag coefficient of $c_d = 0.010655$. Note that SLSQP converges to a similar design (see Figure \ref{fig:det_design}) with a comparable drag coefficient of $c_d = 0.010620$.

\begin{figure}[h]
	\centering
	\begin{footnotesize}
		\includegraphics[width=10.5cm]{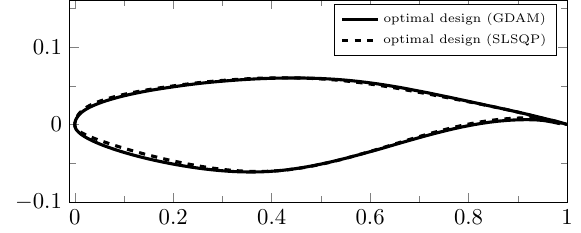}
		\caption{Resulting shapes of deterministic shape optimization problem with GDAM (bold line) and SLSQP (dashed line).}
		\label{fig:det_design}
	\end{footnotesize}
\end{figure}


\subsection{Probabilistic aerodynamic shape optimization of RAE2822 in transonic turbulent flow} 

The empirical mean optimization problem with an uncertain Mach number as introduced in \eqref{eq:EMO} is given by
\begin{equation}
    \begin{split}
        &\min_{x} ~~\mathbb{E}_{\xi \sim \mathcal{N}(0.729, 0.01)}(c_d(x \textcolor{blue}{;} \xi) )  \\
        &~~\text{s.t.} ~~~ c_m(x \textcolor{blue}{;} \bar{\xi}) \leq 0.093, t(x) \geq 0.12.
    \end{split}
    \label{eq:ASSO}
\end{equation}


Based on the algorithm described in Section \ref{sec:algorithms} we find a design by using GDAM in a probabilistic setting. Here, we make use of $16$ samples from $\hat{P}_N$ in each optimization step. In Figure \ref{fig:robust_design}, the resulting design is compared to the design observed for the deterministic optimization. For the optimal design based on the empirical mean, we observe that the trailing edge is thicker. The mean value corresponding to this design, approximated using $200$ samples, is $\mathbb{E}(c_d)=0.011355$. The drag coefficient corresponding to the nominal Mach number is given by $c_d=0.011148$. Note, that the approximated mean value of the design found by the deterministic optimization is $\mathbb{E}(c_d)=0.011527$.

\begin{figure}[h]
	\centering
	\begin{footnotesize}
		\includegraphics[width=10.5cm]{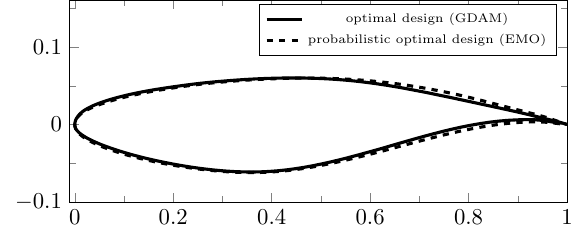}
		\caption{Resulting shapes of deterministic shape optimization problem with GDAM (bold line) and empirical mean optimization problem with stochastic GDAM (dashed line).}
		\label{fig:robust_design}
	\end{footnotesize}
\end{figure}

Details on the performance of the design for different Mach numbers in a given range around the nominal Mach number can be seen in Figure \ref{fig:robust_mach}. While the drag coefficient of the optimal design (bold line) is sensitive around the
Mach number of 0.729, it is not very sensitive for the design found by the probabilistic optimization. The strong
increase of the drag coefficient is also shifted to a higher Mach number. This already shows the advantages of considering an optimization based on the empirical mean. 

\begin{figure}[h]
	\centering
	\begin{footnotesize}
		\includegraphics[width=10.5cm]{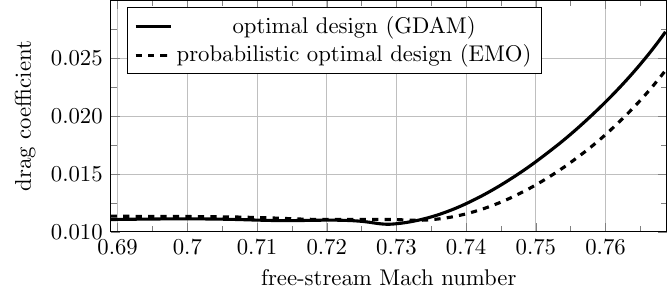}
		\caption{Behaviour of the drag coefficient in a range around the nominal value of the Mach number for the deterministic (bold line) and the probabilistic design (dashed line).}
		\label{fig:robust_mach}
	\end{footnotesize}
\end{figure}


\subsection{Distributionally robust aerodynamic design of RAE2822 in transonic turbulent flow}

We finally apply the two different strategies for distributionally robust optimization using $\chi^2$-divergence as introduced in Section \ref{sec:frameworks} for aerodynamic design. We start with the constrained formulation according to \eqref{eq:DRO} resulting in the optimization problem

\begin{equation}
\begin{split}
&    \min_x   \max_{Q \in \mathdutchcal{D}(\mathcal{N}(0.729, 0.01), \rho)} \mathbb{E}_{\xi \sim Q} \left( c_d(x ; \xi)  \right), \\
& ~\textnormal{s.t.}~~~ c_m(x ; \bar{\xi}) \leq 0.093, t(x) \geq 0.12.
\end{split}
\end{equation}

Given a value of $\rho = \textcolor{blue}{0.5}$, stochastic GDAM is applied to solve the optimization problem using $16$ samples in each step. The resulting design, with a drag coefficient of $c_d = 0.011272$ for the nominal value, is shown in Figure \ref{fig:dro_design} and only small differences to the design based found by optimizing the empirical mean. 

\begin{figure}[h]
	\centering
	\begin{footnotesize}
		\includegraphics[width=10.5cm]{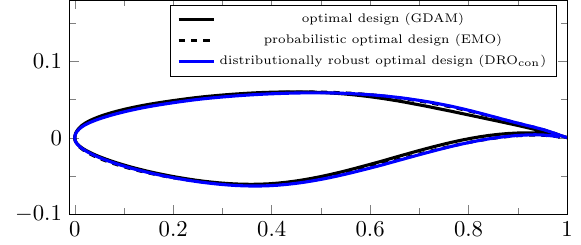}
		\caption{Resulting shapes of deterministic shape optimization (bold line), empirical mean optimization (dashed line) and constrained-based distributionally robust optimization (blue line).}
		\label{fig:dro_design}
	\end{footnotesize}
\end{figure}

Even though the designs are very similar, Figure \ref{fig:dro_mach} shows why the result of the distributionally robust optimization can be considered more robust regarding changes in the Mach number. The behaviour for low Mach numbers is comparable to the result of the probabilistic optimization, while for higher Mach number the increase in the drag coefficient is further reduced. 


\begin{figure}[h]
	\centering
	\begin{footnotesize}
		\includegraphics[width=10.5cm]{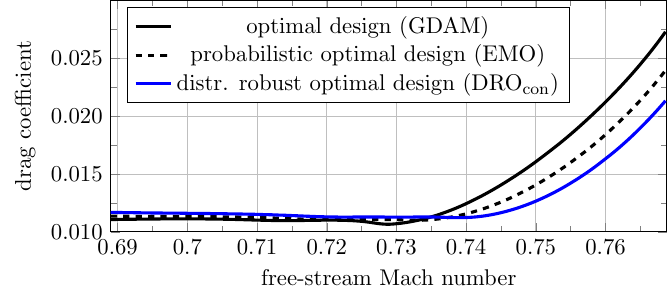}
		\caption{Behaviour of the drag coefficient in a range around the nominal value of the Mach number for the deterministic (bold line), the probabilistic design (dashed line), and the distributionally robust design (blue line).}
		\label{fig:dro_mach}
	\end{footnotesize}
\end{figure}

We also apply stochastic GDAM for solving the distributionally robust optimization problem based on penalization \ref{eq:penalized_DRO}, i.e.,
\begin{equation}
\begin{split}
& \min_x \max_Q \{\mathbb{E}_{\xi \sim Q} \left( c_d(x ;\xi)\right) - \delta d(Q,\mathcal{N}(0.729, 0.01))  \}, \\
& ~\textnormal{s.t.}~~  c_m(x; \bar{\xi}) \leq 0.093, t(x) \geq 0.12
\end{split} 
\end{equation}
with a penalty factor of $\delta= 0.001$. The optimization uses a number of $16$ samples per iteration and converges in about $200$ iterations to a similar design as for the constrained formulation, corresponding to a drag coefficient of $c_d=0.011203$ in the nominal value. Figure \ref{fig:dro_mach_penal} shows that the behaviour for different Mach numbers is also comparable.

\begin{figure}[h]
	\centering
	\begin{footnotesize}
		\includegraphics[width=10.5cm]{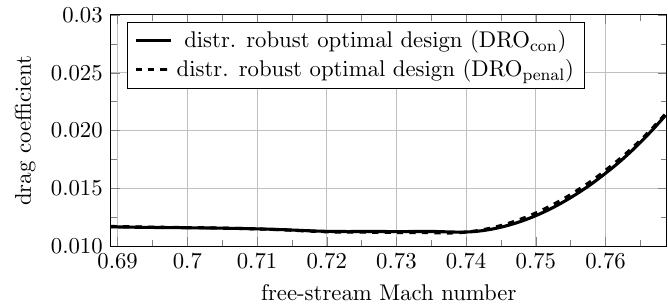}
		\caption{Behaviour of the drag coefficient in a range around the nominal value of the Mach number for the distributionally robust design using the constraint formulation (bold line) and the penalized formulation (dashed line).}
		\label{fig:dro_mach_penal}
	\end{footnotesize}
\end{figure}






\section{Discussions}
We formulate, study, and solve data-driven aerodynamic shape design problems with distributionally robust optimization approaches. Although DRO is a relatively new framework for engineering design, we point out that there is a strong connection between DRO and the robust design method of Taguchi, a proven and successful method in many engineering industries. While the latter is formulated based on the statistical moments of the cost (performance measures), the former is based on observed data with uncertainties (uncontrollable and unknown noise factors). Our preliminary computational experiments on aerodynamic shape design in transonic turbulent flow show promising design results.

\bibliographystyle{tfs}
\bibliography{data_driven_aso}

\end{document}

%% file: fig/multipoint.pdf_tex
\begingroup%
  \makeatletter%
  \providecommand\color[2][]{%
    \errmessage{(Inkscape) Color is used for the text in Inkscape, but the package 'color.sty' is not loaded}%
    \renewcommand\color[2][]{}%
  }%
  \providecommand\transparent[1]{%
    \errmessage{(Inkscape) Transparency is used (non-zero) for the text in Inkscape, but the package 'transparent.sty' is not loaded}%
    \renewcommand\transparent[1]{}%
  }%
  \providecommand\rotatebox[2]{#2}%
  \newcommand*\fsize{\dimexpr\f@size pt\relax}%
  \newcommand*\lineheight[1]{\fontsize{\fsize}{#1\fsize}\selectfont}%
  \ifx\svgwidth\undefined%
    \setlength{\unitlength}{368.50393701bp}%
    \ifx\svgscale\undefined%
      \relax%
    \else%
      \setlength{\unitlength}{\unitlength * \real{\svgscale}}%
    \fi%
  \else%
    \setlength{\unitlength}{\svgwidth}%
  \fi%
  \global\let\svgwidth\undefined%
  \global\let\svgscale\undefined%
  \makeatother%
  \begin{picture}(1,0.61538462)%
    \lineheight{1}%
    \setlength\tabcolsep{0pt}%
    \put(0,0){\includegraphics[width=\unitlength,page=1]{multipoint.pdf}}%
    \put(0.48695265,0.02858717){\color[rgb]{0,0,0}\makebox(0,0)[lt]{\lineheight{1.25}\smash{\begin{tabular}[t]{l}Mach\end{tabular}}}}%
    \put(0.02352483,0.34562575){\color[rgb]{0,0,0}\makebox(0,0)[lt]{\lineheight{1.25}\smash{\begin{tabular}[t]{l}Lift\end{tabular}}}}%
    \put(0,0){\includegraphics[width=\unitlength,page=2]{multipoint.pdf}}%
    \put(0.06418665,0.31339587){\color[rgb]{0,0,0}\makebox(0,0)[lt]{\lineheight{1.25}\smash{\begin{tabular}[t]{l}0.50\end{tabular}}}}%
    \put(0.06384371,0.47158768){\color[rgb]{0,0,0}\makebox(0,0)[lt]{\lineheight{1.25}\smash{\begin{tabular}[t]{l}0.53\end{tabular}}}}%
    \put(0.06598069,0.1591634){\color[rgb]{0,0,0}\makebox(0,0)[lt]{\lineheight{1.25}\smash{\begin{tabular}[t]{l}0.47\end{tabular}}}}%
    \put(0.50486168,0.06821659){\color[rgb]{0,0,0}\makebox(0,0)[lt]{\lineheight{1.25}\smash{\begin{tabular}[t]{l}0.72\end{tabular}}}}%
    \put(0.77829384,0.0671578){\color[rgb]{0,0,0}\makebox(0,0)[lt]{\lineheight{1.25}\smash{\begin{tabular}[t]{l}0.74\end{tabular}}}}%
    \put(0.2327155,0.06855381){\color[rgb]{0,0,0}\makebox(0,0)[lt]{\lineheight{1.25}\smash{\begin{tabular}[t]{l}0.70\end{tabular}}}}%
    \put(0.20712687,0.25549161){\color[rgb]{0,0,0}\makebox(0,0)[lt]{\lineheight{1.25}\smash{\begin{tabular}[t]{l}$w_1 = 0.2$ \end{tabular}}}}%
    \put(0.47565002,0.25601719){\color[rgb]{0,0,0}\makebox(0,0)[lt]{\lineheight{1.25}\smash{\begin{tabular}[t]{l}$w_3 = 0.4$ \end{tabular}}}}%
    \put(0.74263986,0.25629311){\color[rgb]{0,0,0}\makebox(0,0)[lt]{\lineheight{1.25}\smash{\begin{tabular}[t]{l}$w_5 = 0.2$ \end{tabular}}}}%
    \put(0.4754582,0.42202781){\color[rgb]{0,0,0}\makebox(0,0)[lt]{\lineheight{1.25}\smash{\begin{tabular}[t]{l}$w_4 = 0.1$ \end{tabular}}}}%
    \put(0.47443625,0.13051247){\color[rgb]{0,0,0}\makebox(0,0)[lt]{\lineheight{1.25}\smash{\begin{tabular}[t]{l}$w_2 = 0.1$ \end{tabular}}}}%
  \end{picture}%
\endgroup%

%% file: fig/flight_data_points_hel_muc.pdf_tex
\begingroup%
  \makeatletter%
  \providecommand\color[2][]{%
    \errmessage{(Inkscape) Color is used for the text in Inkscape, but the package 'color.sty' is not loaded}%
    \renewcommand\color[2][]{}%
  }%
  \providecommand\transparent[1]{%
    \errmessage{(Inkscape) Transparency is used (non-zero) for the text in Inkscape, but the package 'transparent.sty' is not loaded}%
    \renewcommand\transparent[1]{}%
  }%
  \providecommand\rotatebox[2]{#2}%
  \newcommand*\fsize{\dimexpr\f@size pt\relax}%
  \newcommand*\lineheight[1]{\fontsize{\fsize}{#1\fsize}\selectfont}%
  \ifx\svgwidth\undefined%
    \setlength{\unitlength}{340.15748031bp}%
    \ifx\svgscale\undefined%
      \relax%
    \else%
      \setlength{\unitlength}{\unitlength * \real{\svgscale}}%
    \fi%
  \else%
    \setlength{\unitlength}{\svgwidth}%
  \fi%
  \global\let\svgwidth\undefined%
  \global\let\svgscale\undefined%
  \makeatother%
  \begin{picture}(1,0.75)%
    \lineheight{1}%
    \setlength\tabcolsep{0pt}%
    \put(0.48393691,0.00951001){\color[rgb]{0,0,0}\makebox(0,0)[lt]{\lineheight{1.25}\smash{\begin{tabular}[t]{l}Mach\end{tabular}}}}%
    \put(0.09001818,0.04598143){\color[rgb]{0,0,0}\makebox(0,0)[lt]{\lineheight{1.14999998}\smash{\begin{tabular}[t]{l}0.66\end{tabular}}}}%
    \put(0.22593184,0.04534834){\color[rgb]{0,0,0}\makebox(0,0)[lt]{\lineheight{1.14999998}\smash{\begin{tabular}[t]{l}0.68\end{tabular}}}}%
    \put(0.36491211,0.04510619){\color[rgb]{0,0,0}\makebox(0,0)[lt]{\lineheight{1.14999998}\smash{\begin{tabular}[t]{l}0.70\end{tabular}}}}%
    \put(0.50134735,0.0448796){\color[rgb]{0,0,0}\makebox(0,0)[lt]{\lineheight{1.14999998}\smash{\begin{tabular}[t]{l}0.72\end{tabular}}}}%
    \put(0.63677656,0.04500413){\color[rgb]{0,0,0}\makebox(0,0)[lt]{\lineheight{1.14999998}\smash{\begin{tabular}[t]{l}0.74\end{tabular}}}}%
    \put(0.77401585,0.04685736){\color[rgb]{0,0,0}\makebox(0,0)[lt]{\lineheight{1.14999998}\smash{\begin{tabular}[t]{l}0.76\end{tabular}}}}%
    \put(0.90654071,0.04985935){\color[rgb]{0,0,0}\makebox(0,0)[lt]{\lineheight{1.14999998}\smash{\begin{tabular}[t]{l}0.78\end{tabular}}}}%
    \put(0,0){\includegraphics[width=\unitlength,page=1]{flight_data_points_hel_muc.pdf}}%
    \put(0.01665051,0.40135412){\color[rgb]{0,0,0}\makebox(0,0)[lt]{\lineheight{1.14999998}\smash{\begin{tabular}[t]{l}Lift\end{tabular}}}}%
    \put(0.06239118,0.6675114){\color[rgb]{0,0,0}\makebox(0,0)[lt]{\lineheight{1.14999998}\smash{\begin{tabular}[t]{l}0.56\end{tabular}}}}%
    \put(0.06291092,0.57454345){\color[rgb]{0,0,0}\makebox(0,0)[lt]{\lineheight{1.14999998}\smash{\begin{tabular}[t]{l}0.54\end{tabular}}}}%
    \put(0.06237412,0.47055523){\color[rgb]{0,0,0}\makebox(0,0)[lt]{\lineheight{1.14999998}\smash{\begin{tabular}[t]{l}0.52\end{tabular}}}}%
    \put(0.06367728,0.36893624){\color[rgb]{0,0,0}\makebox(0,0)[lt]{\lineheight{1.14999998}\smash{\begin{tabular}[t]{l}0.50\end{tabular}}}}%
    \put(0.06220058,0.2685793){\color[rgb]{0,0,0}\makebox(0,0)[lt]{\lineheight{1.14999998}\smash{\begin{tabular}[t]{l}0.48\end{tabular}}}}%
    \put(0.06191549,0.16466954){\color[rgb]{0,0,0}\makebox(0,0)[lt]{\lineheight{1.14999998}\smash{\begin{tabular}[t]{l}0.46\end{tabular}}}}%
    \put(0.06416944,0.06751227){\color[rgb]{0,0,0}\makebox(0,0)[lt]{\lineheight{1.14999998}\smash{\begin{tabular}[t]{l}0.44\end{tabular}}}}%
    \put(0,0){\includegraphics[width=\unitlength,page=2]{flight_data_points_hel_muc.pdf}}%
  \end{picture}%
\endgroup%